\documentclass[useAMS]{tCON2e}
\usepackage{natbib}
\begin{document}
\doi{10.1080/00207170xxxxxxxxxxxx}
 \issn{1366-5820}
\issnp{0020-7179} 






\newcommand{\dst}{\displaystyle}

\def\bbc{{\mathbb C}}
\def\bbr{{\mathbb R}}
\newcommand{\be}{\begin{equation}}
\newcommand{\ee}{\end{equation}}
\newcommand{\bd}{\begin{displaymath}}
\newcommand{\ed}{\end{displaymath}}
\newcommand{\bea}{\begin{eqnarray}}
\newcommand{\eea}{\end{eqnarray}}
\newcommand{\R}{\mathbb{R}}
\newcommand{\Rp}{\mathbb{R}_{+}}
\newcommand{\Hi} {{\cal H}^\infty}
\newcommand{\RHi}{{\cal RH}^\infty}
\newcommand{\rhp}{\mathbb{C}_{+}}
\newcommand{\Zp}{\mathbb{Z}_{+}}
\newcommand{\Znn}{\mathbb{Z}_{\geq0}}

\markboth{S.~Gumussoy and H.~\"Ozbay}{S.~Gumussoy and H.~\"Ozbay}

\title{Stable $\Hi$ controller design for time-delay
systems$^\dagger$\footnote{$\dagger$ This work is supported in part
by the European Commission under contract No.~MIRG-CT-2004-006666,
and by T\"UB\.{I}TAK under grant numbers EEEAG-105E065 and
EEEAG-105E156}}

\author{S.~Gumussoy$^\ddagger$ and H.~\"Ozbay\thanks{$^\ast$Corresponding author. Email: hitay@bilkent.edu.tr or
ozbay.1@osu.edu\vspace{10pt}
\newline\centerline{\tiny{ {\em International Journal of
Control}}}
\newline\centerline{\tiny{ISSN 0020-7179 print/ ISSN
1366-5820
 online
\textcopyright 2005 Taylor \& Francis Ltd}}
\newline\centerline{\tiny{ http://www.tandf.co.uk/journals}}
\newline \centerline{\tiny{DOI:
10.1080/00207170xxxxxxxxxxxx}}}$^\ast$$^\S$\\\vspace{6pt}
$^\ddagger$MIKES Inc., Akyurt, ANKARA TR-06750, Turkey\\
$^\S$Dept. of Electrical and Electronics Eng., Bilkent University,
Bilkent, Ankara TR-06800, Turkey, \\on leave from Dept. of
Electrical and Computer Eng., Ohio State University, Columbus, OH
43210} \received{\today}

\maketitle
\begin{abstract}
This paper investigates stable suboptimal $\Hi$ controllers for a
class of single-input single-output time-delay systems. For a given
plant and weighting functions, the optimal controller minimizing the
mixed sensitivity (and the central suboptimal controller) may be
unstable with finitely or infinitely many poles in $\rhp$.  For each
of these cases search algorithms are proposed to find stable
suboptimal $\Hi$ controllers. These design methods are illustrated
with examples.
\end{abstract}

\section{Introduction}

In a feedback system, stable stabilizing controllers (also called
strongly stabilizing controllers) are desired for many practical
reasons, \cite{V85}. It is shown that \cite{YBL-AUT-74,Abedor89}
such controllers can be designed if and only if the plant satisfies
the parity interlacing property.  A design method for finding
strongly stabilizing controllers for SISO plants with input-output
(I/O) time delays is given in \cite{Suyama91} where a stable
controller is constructed by finding a unit (an outer function whose
inverse is proper) satisfying certain interpolation conditions.

In the literature, stable controllers satisfying a performance
requirement are also studied. For example, design methods are given
for $\Hi$ strong stabilization for finite dimensional plants, see
e.g.
\cite{Sideris,Ganesh,Jacobus,Ito,Barabanov,ZO99,ZO00,CC-TAC-01,Zhou01,LS-SCL-02,CZ-TAC-03,CWL-CDC-03}
and their references. For time delay systems, $\Hi$-based strong
stabilization is also considered. Optimal stable $\Hi$ controller
design for a class of SISO time-delay systems within the framework
of the weighted sensitivity minimization problem is studied in
\cite{GO-ICCARV-06}. It is known that $\Hi$ controllers for
time-delay systems with finitely many unstable poles can be designed
by the methods in \cite{FTZ86,ZK87,TO95,GO-IFAC-06}. For this class
of plants, weighted sensitivity problem may result in an optimal
$\Hi$ controller with infinitely unstable modes, \cite{Flamm,Lenz}.
For the mixed sensitivity minimization problem, an indirect approach
to design a stable controller achieving a desired $\Hi$ performance
level for finite dimensional SISO plants with I/O delays is proposed
in \cite{Gumussoy02}. This approach is based on stabilization of the
unstable optimal, or central suboptimal, $\Hi$ controller by another
$\Hi$ controller in the feedback loop. In \cite{Gumussoy02},
stabilization is achieved and the sensitivity deviation is minimized
under certain sufficient conditions. There are two main drawbacks of
this method. First, the solution of sensitivity deviation brings
conservatism because of finite dimensional approximation of the
infinite dimensional weight. Second, the stability of overall
sensitivity function is
not guaranteed. 

In \cite{GO-MTNS-04} we focused on strong stabilization problem for
SISO plants with I/O delays  such that the stable controller
achieves the pre-specified suboptimal $\Hi$ performance level in the
mixed sensitivity minimization problem. When the optimal controller
is unstable (with infinitely or finitely many unstable poles), two
methods are given based on a search algorithm to find a stable
suboptimal controller. However, both methods are conservative. In
other words, there may be a stable suboptimal controller achieving a
smaller performance level. In \cite{GO-MTNS-04} necessary conditions
for stability of optimal and suboptimal controllers are also given.

In this paper, the results of \cite{GO-MTNS-04} are extended for
general SISO time-delay systems in the form \be \label{eq:preplant}
P(s)=\frac{r_p(s)}{t_p(s)}=\frac{\sum_{i=1}^{n}r_{p,i}(s)e^{-h_i
s}}{\sum_{j=1}^{m}t_{p,j}(s)e^{-\tau_j s}} \ee satisfying the
assumptions
\begin{enumerate}
    \item[A.1]
    \begin{enumerate}
    \item $r_{p,i}(s)$ and $t_{p,j}(s)$ are polynomials with
    real coefficients;
    \item $\{h_i\}_{i=1}^n$ and $\{\tau_i\}_{i=1}^m$
    are two sets of strictly increasing nonnegative
    rational numbers with $h_1\geq\tau_1$;
    \item define the polynomials $r_{p,i_{max}}$ and
    $t_{p,j_{max}}$ with largest
    polynomial degree in $r_{p,i}$ and $t_{p,j}$ respectively
    (the smallest index if there is more than one),
    then, ${\textrm{deg}\{r_{p,i_{max}}(s)\}}
    \leq {\textrm{deg}\{t_{p,{j_{max}}}(s)\}}$
    and
    $h_{i_{max}}\geq\tau_{j_{max}}$ where $\textrm{deg}\{.\}$
    denotes the degree of the polynomial;
    \end{enumerate}
    \item[A.2] $P$ has no imaginary axis zeros or poles;
    \item[A.3] $P$ has finitely many unstable poles, or
    equivalently $t_{p}(s)$ has finitely many zeros in
    $\rhp$;
    \item[A.4] $P$ can be written in the form of
    \be \label{eq:mnfac}
    P(s)=\frac{m_n(s)N_o(s)}{m_d(s)}
    \ee
    where $m_n$, $m_d$ are inner, infinite and finite dimensional,
    respectively;
    $N_o$ is outer, possibly infinite dimensional as in \cite{TO95}.
\end{enumerate}
Conditions stated in $A.1$ are not restrictive, and in most cases
$A.2$ can be removed if the weights are chosen in a special manner.
The conditions $A.3-A.4$ come from the restrictions of the
Skew-Toeplitz approach to $\Hi$-control of infinite dimensional
systems.
It is not easy to check assumptions $A.3-A.4$, unless a
quasi-polynomial root finding algorithm is used. In Section
\ref{sec:Pfact}, we will give a necessary and sufficient condition
to check the assumption $A.3$.

The optimal $\Hi$ controller, $C_{opt}$, stabilizes the feedback
system and achieves the minimum $\Hi$ cost,~$\gamma_{opt}$: \be
\label{eq:wsm} \gamma_{opt} = \left\| \left[
\begin{array}{c}
  W_1(1+PC_{opt})^{-1} \\
  W_2PC_{opt}(1+PC_{opt})^{-1}
\end{array} \right] \right\|_\infty
= \inf_{C \; stab. \; P} \left\| \left[ \begin{array}{c}
  W_1(1+PC)^{-1} \\
  W_2PC(1+PC)^{-1}
\end{array} \right] \right\|_\infty
\ee where $W_1$ and $W_2$ are finite dimensional weights for this
mixed sensitivity minimization problem.

In Section \ref{sec:Pfact}, conditions are given to check
assumptions $A.3$ and $A.4$, and an algorithm is derived for the
plant factorization (\ref{eq:mnfac}). Section \ref{sec:Cstruct}
discusses the structure of optimal and suboptimal $\Hi$ controllers.
Stable suboptimal $\Hi$ controller design methods for the cases
where the optimal controller has infinitely or finitely many
unstable poles are discussed in Sections \ref{sec:infinitelymany}
and \ref{sec:finitelymany} respectively. Examples can be found in
Section \ref{sec:examp}, and concluding remarks are made in Section
\ref{sec:concl}.

\textbf{Definition:} A function $F(s)$ defined on right half of
complex plane is called {\it proper} (respectively {\it strictly
proper}) if \bd \lim_{|s|\rightarrow\infty}|F(s)|<\infty \quad
\left(\textrm{respectively}
\lim_{|s|\rightarrow\infty}|F(s)|=0\right). \ed The function is
called {\it biproper} if the limit converges to a nonzero value.

\section{Assumptions and Factorization of Plant } \label{sec:Pfact}

Note that by multiplying and dividing (\ref{eq:preplant}) by a
stable polynomial, it is always possible to put the plant in the
form \be \label{eq:seqplant}
P(s)=\frac{R(s)}{T(s)}=\frac{\sum_{i=1}^n R_i(s)e^{-h_i
s}}{\sum_{j=1}^m T_j(s)e^{-\tau_j s}} \ee where $R_i$ and $T_j$ are
finite dimensional, stable, proper transfer functions.   In this
section, we study conditions to verify assumptions $A.3$ and $A.4$.

\begin{lemma} (\cite{GO-IFAC-06}) \label{lemma:finitezeros}
Assume that $R(s)$ in (\ref{eq:seqplant}) has no imaginary axis
zeros and poles, then the system, $R$, has finitely many unstable
zeros if and only if all the  roots of the polynomial,
$\varphi(r)=1+\sum_{i=2}^n \xi_i r^{\tilde{h}_i-\tilde{h}_1}$
 has magnitude greater than $1$ where
\bea \nonumber  \xi_i
 &=&\lim_{\omega\rightarrow\infty}R_i(j\omega)R_1^{-1}(j\omega)
 \quad \forall \;i=2,\dots,n, \\
\nonumber h_i&=&\frac{\tilde{h}_i}{N}, \quad N, \tilde{h}_i \in\Zp,
\; \forall \;i=1,\dots,n. \eea
\end{lemma}

We define the conjugate of $R(s)=\sum_{i=1}^n R_i(s)e^{-h_i s}$ in
(\ref{eq:seqplant}) as $\bar{R}(s):=e^{-h_ns}R(-s)M_C(s)$ where
$M_C$ is inner, finite dimensional whose poles are the poles of $R$.
If the time delay system $R$ has finitely many $\rhp$ zeros it is
called an {\it $F$-system}. It is clear that $R$ is an $F$-system if
it satisfies Lemma \ref{lemma:finitezeros}.  If the time delay
system $\bar{R}$ has finitely many $\rhp$ zeros then $R$ is said to
be an {\it $I$-system}.

\begin{corollary} (\cite{GO-IFAC-06}) \label{cor:A3A4}
The plant $P=\frac{R}{T}$ in (\ref{eq:seqplant}) satisfies $A.3-A.4$
if one of the following conditions hold: (i) $R$ is $I$-system and
$T$ is $F$-system, or (ii) $R$ and $T$ are $F$-systems with
$h_1>\tau_1$.
\end{corollary}

In \cite{GO-IFAC-06}, it is shown that the plant factorization
(\ref{eq:seqplant})  can be done as (\ref{eq:mnfac}) when
\begin{itemize}
    \item[(i)] $R$ is an $I$-system and $T$ is an $F$-system,
    \bea \label{eq:IFfactorization} \nonumber
m_n&=&e^{-(h_1-\tau_1)s}M_{\bar{R}} \frac{(e^{h_1
s}R)}{\bar{R}}, \\
\nonumber m_d&=&M_T, \\
N_o&=&\frac{\bar{R}}{M_{\bar{R}}} \frac{M_T}{(e^{\tau_1 s}T)}, \eea
    \item[(ii)] $R$ and $T$ are $F$-systems with $h_1>\tau_1$,
        \bea \label{eq:FFfactorization} \nonumber
m_n&=&e^{-(h_1-\tau_1)s}M_R, \\
\nonumber m_d&=&M_T, \\
N_o&=&\frac{R}{M_R} \frac{M_T}{(e^{\tau_1 s}T)} \eea
\end{itemize}

\noindent where $M_R$ and $M_{\bar{R}}$ are inner functions whose
zeros are the $\rhp$ zeros of $R$ and $\bar{R}$ respectively. When
$R$ is an $I$-system, conjugate of $R$  has finitely many unstable
zeros, so $M_{\bar{R}}$ is well-defined. Similarly, zeros of $M_T$
are unstable zeros of $T$. Note that $m_n$ and $m_d$ are inner
functions, infinite and finite dimensional respectively. The
function $N_o$ is outer. By (\ref{eq:FFfactorization}), one can see
that the condition $h_1>\tau_1$ is necessary for $m_n$ to be a
causal and infinite dimensional system. For further details, see
\cite{GO-IFAC-06}.

\section{Structure of $\Hi$ Controllers} \label{sec:Cstruct}

Assume that the problem data in (\ref{eq:wsm}) satisfies that $W_1$
is non-constant function and $(W_2 N_o), (W_2 N_o)^{-1}\in\Hi$, then
the optimal $\Hi$ controller can be written as, \cite{TO95}, \be
\label{eq:Copt} C_{opt}=E_{\gamma_{opt}}
m_d\frac{N_o^{-1}F_{\gamma_{opt}}L}{1+m_nF_{\gamma_{opt}}L} \ee
where $E_{\gamma}=\left(\frac{W_1(-s)W_1(s)}{\gamma^2}-1\right)$,
and for the definition of the other terms, let the right half plane
zeros of $E_\gamma$ be $\beta_i$, $i=1,\ldots,n_1$, the right half
plane poles of $P$ be $\alpha_i$, $i=1,\ldots,\ell$ and that of
$W_1(-s)$ be $\eta_i$ for $i=1,\ldots,n_1$. Then,
$F_{\gamma}(s)=G_{\gamma}(s)\prod_{i=1}^{n_1}\frac{s-\eta_i}{s+\eta_i}$
where \bea \label{eq:spectralfactorization}
G_{\gamma}(s)G_{\gamma}(-s)
=\left(1-\left(\frac{W_2(-s)W_2(s)}{\gamma^2}-1\right)E_\gamma
\right)^{-1} \eea and $G_\gamma\in\Hi$ is outer function. The
rational function $L=\frac{L_2}{L_1}$ , $L_1$ and $L_2$ are
polynomials with degrees less than or equal to $(n_1+\ell-1)$ and
they are determined by the following interpolation conditions, \bea
\label{eq:interpcond}
0&=&L_1(\beta_i)+m_n(\beta_i)F_{\gamma}(\beta_i)L_2(\beta_i),\\
\nonumber 0&=&L_1(\alpha_k)+m_n(\alpha_k)F_{\gamma}(\alpha_k)L_2(\alpha_k),\\
\nonumber 0&=&L_2(-\beta_i)+m_n(\beta_i)F_{\gamma}(\beta_i)L_1(-\beta_i),\\
\nonumber
0&=&L_2(-\alpha_k)+m_n(\alpha_k)F_{\gamma}(\alpha_k)L_1(-\alpha_k)
\eea for $i=1,\ldots,n_1$ and $k=1,\ldots,\ell$. The optimal
performance level, $\gamma_{opt}$, is the largest $\gamma$ value
such that spectral factorization (\ref{eq:spectralfactorization})
exists and interpolation conditions (\ref{eq:interpcond}) are
satisfied.

Similarly, all suboptimal controllers achieving the performance
level $\rho>\gamma_{opt}$ can be written as, \cite{TO95}, \be
\label{eq:Csubopt} C_{subopt}=E_{\rho}
m_d\frac{N_o^{-1}F_{\rho}L_U}{1+m_nF_{\rho}L_U} \ee where
$\rho>\gamma_{opt}$ and
$L_U(s)=\frac{L_{2U}}{L_{1U}}=\frac{L_2(s)+L_1(-s)U(s)}{L_1(s)+L_2(-s)U(s)}$
with $U\in\Hi$, $\|U\|_\infty\leq1$. The polynomials, $L_1$ and
$L_2$, have degrees less than or equal to $n_1+\ell$. Same
interpolation conditions (\ref{eq:interpcond}) are valid with $\rho$
replacing $\gamma$. Moreover, there are two additional conditions on
$L_1$ and $L_2$: \bea
\nonumber 0&=&L_2(-a)+(E_\rho(a)+1)F_\rho(a)m_n(a)L_1(-a) \\
\nonumber 0&\neq&L_1(-a) \eea where $a\in\Rp$ is arbitrary.

Note that the $\rhp$ zeros of $E_{\gamma_{opt}}$ and $m_d$ are
always cancelled by the denominator in (\ref{eq:Copt}). Therefore,
$C_{opt}$ is stable if and only if the denominator in
(\ref{eq:Copt}) has no zeros in $\rhp$ except the zeros of
$E_{\gamma_{opt}}$ and $m_d$ in $\rhp$ (multiplicities considered).
Same conclusion is valid for the suboptimal case.

\begin{lemma} \label{lemma:infcond}
Let the plant  (\ref{eq:seqplant}) satisfy $A.1$-$A.4$. The optimal
controller for the mixed sensitivity problem (\ref{eq:wsm}),  and
respectively a suboptimal controller with finite dimensional $U$,
have infinitely many poles in $\rhp$ if and only if the following
inequalities hold respectively, \bea \label{eq:inequalitiesinfcond}
\nonumber \lim_{\omega\rightarrow\infty}
|F_{\gamma_{opt}}(j\omega)L_{opt}(j\omega)|
&\geq&1  \\
\lim_{\omega\rightarrow\infty}|F_{\rho}(j\omega)L_U(j\omega)|&\geq&1.
\eea
\end{lemma}
\begin{proof}
The optimal (respectively suboptimal) controller has
infinitely many  poles in $\rhp$ if and only if the equations \bea
\nonumber \label{eq:1mnFL} 1+m_n(s)F_{\gamma_{opt}}(s)L_{opt}(s)
&=&0 \quad\textrm{respectively}, \\
 \quad 1+m_n(s)F_{\rho}(s)L_U(s)&=&0.
\eea have infinitely many roots in $\rhp$. Assume that the Nyquist
contour in right-half plane is chosen such that the $\rhp$ zeros of
$E_{\gamma_{opt}}$ (resp. $E_\rho$) and $m_d$ are excluded. The
unstable poles of the term (\ref{eq:1mnFL}) are the unstable poles
of $L_{opt}$ (resp. $L_U$) which are finitely many (note that $L_2$,
$L_1$ and $U$ are finite dimensional). Using Nyquist theorem, we can
conclude that the term (\ref{eq:1mnFL}) has infinitely many zeros in
$\rhp$ if and only if Nyquist plot of $m_n F_{\gamma_{opt}}L_{opt}$
(resp. $m_nF_\rho L_U$) encircles $-1$ infinitely many times. This
is equivalent to the following conditions: \bea \nonumber
\lim_{\omega\rightarrow\infty}
|F_{\gamma_{opt}}(j\omega)L_{opt}(j\omega)|
&\geq&1  \quad\textrm{respectively}, \\
\nonumber \lim_{\omega\rightarrow\infty}
|F_{\rho}(j\omega)L_U(j\omega)| &\geq&1 \eea and $m_n$ encircles the
origin infinitely many times. When $R$ is an $I$-system and $T$ is
an $F$-system, $m_n$ has infinitely many zeros in $\rhp$ and no
poles in $\rhp$, so it encircles the origin infinitely many times.
On the other hand, when $R$ and $T$ are $F$-systems with
$h_1>\tau_1$, we have $m_n=e^{-(h_1-\tau_1)s}M_R$ (where $M_R$ is
finite dimensional), so $m_n$ encircles the origin infinitely times
due to the delay term. Therefore, the inequalities are necessary and
sufficient conditions for controller to have infinitely many
unstable poles.
\end{proof}

The following result gives a necessary and sufficient condition for
a suboptimal controller to have finitely many unstable poles.

\begin{corollary} \label{cor:infitycond}
Let the plant (\ref{eq:seqplant}) satisfy  $A.1$-$A.4$. Assume that
the optimal controller of mixed sensitivity problem (\ref{eq:wsm})
has infinitely many unstable poles. When $U$ is finite dimensional,
the suboptimal controller has finitely many unstable poles if and
only if \be \label{eq:inftycond}
\lim_{\omega\rightarrow\infty}|F_\rho(j\omega)L_U(j\omega)|<1 \ee
\end{corollary}

When the optimal controller has infinitely many unstable poles, a
stable suboptimal controller may be found by proper selection of the
free parameter $U$. In Section \ref{sec:infinitelymany}, this case
is considered.

When $F_{\gamma_{opt}}$ is strictly proper, then the optimal and
suboptimal controllers always have finitely many unstable poles.
Existence condition for strictly proper $F_{\gamma_{opt}}$ and
stable suboptimal $\Hi$ controller design for this case is given in
Section \ref{sec:finitelymany}.

\section{Stable Suboptimal $\Hi$ Controller Design when the Optimal
Controller has Infinitely Many Poles in $\rhp$}
\label{sec:infinitelymany}

Corollary \ref{cor:infitycond} gives a condition on the problem data
so that the suboptimal $\Hi$ controller (which is uniquely
determined by $U$) has finitely many poles in $\rhp$. This condition
will be used to determine a parameter range of $U$. Assume that
$U(s)$ is finite dimensional and bi-proper, and define \bea
\nonumber
f_\infty&:=&\lim_{\omega\rightarrow\infty}|F_\rho(j\omega)|>0, \\
\nonumber u_\infty&:=&\lim_{\omega\rightarrow\infty}U(j\omega)
\quad\textrm{and} \quad u_\infty\in[-1,1], \\
\nonumber
k&:=&\lim_{\omega\rightarrow\infty}\frac{L_{2}(j\omega)}{L_{1}(j\omega)}.
\eea
\begin{lemma} \label{eq:inftyvalues}
Consider the set of suboptimal controllers for the plant
(\ref{eq:seqplant}) with a given $\Hi$ performance level
$\rho>\gamma_{opt}$. This set contains an element with finitely many
poles in $\rhp$ if and only if one of the following conditions is
satisfied: (i) $|k|<1$, or (ii) $|k|\geq1$ and $f_\infty<1$. The
corresponding intervals for $u_\infty$ resulting a suboptimal
controller with finitely many $\rhp$ poles are
\begin{itemize}
    \item[] (i) $(-1)^{n_1+\ell}u_\infty\in\left[-1,1\right]\bigcap
    \left(-\frac{1+f_\infty k}{f_\infty+k},\frac{1-f_\infty
    k}{|f_\infty-k|}\right)$, when $|k|<1$,
    \item[] (ii) $(-1)^{n_1+\ell}u_\infty\in\left[-1,-\frac{1+f_\infty
k}{f_\infty+k}\right)\bigcup\left(\frac{1-f_\infty
    k}{|f_\infty-k|},1\right]$ when $|k|>1$ and $f_\infty<1$
    and $u_\infty\in[-1,1]$ when $|k|=1$ and $f_\infty<1$,
\end{itemize}
where $n_1$ is the dimension of the sensitivity weight $W_1$ and
$\ell$ is the number of $\rhp$  poles of the plant (\ref{eq:mnfac}).
\end{lemma}
\begin{proof}
Using Lemma \ref{lemma:infcond}, there exists suboptimal controller
with finitely many unstable poles if and only if the following
inequality is satisfied, \bd -\frac{1}{f_\infty}<
\frac{k+\tilde{u}_\infty}{1+k\tilde{u}_\infty}<\frac{1}{f_\infty},
\ed where $\tilde{u}_\infty=(-1)^{n_1+\ell}u_\infty$ and
$\tilde{u}_\infty\in[-1,1]$. After algebraic manipulations, one can
see that the admissible $\tilde{u}_\infty$ intervals are
\begin{itemize}
    \item[(a)] $\tilde{u}_\infty\in
    \left(-\frac{1+f_\infty k}{f_\infty+k},\frac{1-f_\infty
    k}{|f_\infty-k|}\right)$ when $f_\infty\geq1$ and $|k|<1$,
    \item[(b)] $\tilde{u}_\infty\in[-1,1]$ when  $f_\infty<1$ and
    $|k|<1$,
    \item[(c)] $\tilde{u}_\infty\in
    \left[-1,-\frac{1+f_\infty k}{f_\infty+k}\right)
    \bigcup\left(\frac{1-f_\infty k}{|f_\infty-k|},1\right]$
    when $|k|>1$ and $f_\infty<1$,
    \item[(d)] $\tilde{u}_\infty\in[-1,1]$ when $|k|=1$ and
    $f_\infty<1$.
\end{itemize} The intervals for admissible  $u_\infty$ in
$(i)$ and $(ii)$ are the results of (a-b) and (c-d) respectively.
This result is a generalized version of a similar result we
presented in \cite{GO-MTNS-04}.
\end{proof}

Note that $u_\infty$ is a design parameter and a valid range to have
a stable $\Hi$ controller can be calculated by $f_\infty$ and $k$.

\begin{theorem} \label{eq:thmss}
Let the plant (\ref{eq:seqplant}) satisfy  $A.1$-$A.4$. Assume that
the optimal and the central suboptimal (for $\rho>\gamma_{opt}$)
controllers determined from the mixed sensitivity problem have
infinitely many unstable poles. If there exists $U\in\Hi$,
$\|U\|_\infty<1$ such that $L_{1U}$ has no $\rhp$ zeros and \be
\label{eq:magcond} |L_U(j\omega)F_\rho(j\omega)|<1, \quad \forall \;
\omega\in[0,\infty), \ee then the suboptimal controller is stable.
\end{theorem}
\begin{proof}
Assume that there exists $U$ satisfying the conditions of the
theorem. By maximum modulus theorem, \bd
|1+m_n(s_o)F_\rho(s_o)L_U(s_o)|>1-|F_\rho(j\omega)L_U(j\omega)|>0,
\ed therefore, there is no unstable zero, $s_o=\sigma+j\omega$ with
$\sigma>0$. The suboptimal controller has no unstable poles.
\end{proof}

Note that Theorem \ref{eq:thmss} is a conservative result and the
level of conservatism can be analyzed case by case with examples.
Although the inequality (\ref{eq:magcond}) is not satisfied, the
term $(1+m_n F_\rho L_U)^{-1}$ can stable. It is difficult to
characterize all $U$ which makes $(1+m_nF_\rho L_U)^{-1}$ stable.
Therefore, the following algorithm tries to find stable controllers
even if the inequality is not satisfied by choosing suitable
$\omega_{max}$ and $\eta_{max}$.

The theorem does not give a systematic method for calculating $U$
which results in a stable $\Hi$ controller.
In order to  address this issue, at least partially, we will
consider the use of first order bi-proper $U$.
Define 
\bea \nonumber \omega_{max}&=&
\max \{\omega ~:~|L_U(j\omega)F_\rho(j\omega)|=1\}, \\
\nonumber
\eta_{max}&=&\max_{\omega\in[0,\infty)}{|L_U(j\omega)F_\rho(j\omega)|}.
\eea Clearly, the choice of $U$ should be such that $\omega_{max}$
and $\eta_{max}$ are as small as possible.
The design method given below searches for a suitable first order
$U$.

\noindent \textbf{Algorithm}\\
Define $U(s)=u_\infty\left(\frac{u_z+s}{u_p+s}\right)$ such that
$u_\infty,u_p,u_z\in\R$, $|u_\infty|\leq1$, $u_p>0$ and $u_p\geq
|u_\infty u_z|$,
\begin{enumerate}
  \item[1)] Fix $\rho>\gamma_{opt}$,
  \item[2)] Calculate $f_\infty$ and $k$,
  \item[3)] Calculate admissible values of $u_\infty$ by using
  Lemma \ref{eq:inftyvalues},
  if no admissible value exists, increase $\rho$ and go back to step 2,
  \item[4)]  Search admissible values for ($u_\infty$, $u_p$, $u_z$) such that
  $L_{1U}(s)$ is stable, if no admissible value exists, increase $\rho$ and
  go back to step 2,
  \item[5)] Find the triplet, $(u^o_{\infty}, u^o_z, u^o_p)$
  minimizing $\omega_{max}$ and
  $\eta_{max}$ for all admissible $(u_\infty,u_p,u_z)$.
  \item[6)] Take a Nyquist contour including the region
  ${\mathcal D}=\{s\in\rhp:
  |m_n(s)F_\rho(s)L_U(s)|>1\}$ (excluding the singularities on imaginary axis).
  Obtain Nyquist plot of $m_n F_\rho L_U $. If the number of encirclement
  of $-1$ is equal to
  unstable zeros of $E_\rho$ and $m_d$ (except the zeros on imaginary axis),
  the
  $\Hi$ controller is stable for
  $U(s)=u^o_\infty\left(\frac{s+u^o_z}{s+u^o_p}\right)$.
  Otherwise, increase $\rho$ and go back to step 2.
\end{enumerate}

When the central suboptimal controller has infinitely many $\rhp$
poles, it is not possible to obtain a stable suboptimal controller
by using a strictly proper or inner $U$. Once we find $U$ from the
above algorithm, the resulting suboptimal stable $\Hi$ controller
can be represented as cascade and feedback connections containing
finite impulse response filter that does not have unstable pole-zero
cancellations in the controller, as explained in \cite{GO-IFAC-06}.
This rearrangement eliminates unstable pole-zero cancellations in
the controller and makes the a practical implementation of the
controller feasible.

\section{Stable Suboptimal $\Hi$ Controller Design when the Optimal
Controller has Finitely Many Poles in $\rhp$}
\label{sec:finitelymany}

In this section, we will give a condition for $\Hi$ controllers to
have finitely many unstable poles. A sufficient condition for the
existence of stable suboptimal $\Hi$ controllers is given, and a
design method is proposed.

The optimal and suboptimal controllers have infinitely many unstable
poles if and only if the inequalities (\ref{eq:inequalitiesinfcond})
are satisfied. On the other hand, the $\Hi$ controllers have always
finitely many unstable poles regardless of problem data if
$F_{\gamma_{opt}}$ and $F_\rho$ are strictly proper. The following
Lemma gives a necessary and sufficient condition when
$F_{\gamma_{opt}}$ and $F_\rho$ are strictly proper.
\begin{lemma}
The $\Hi$ controller has finitely many unstable poles if the plant
is strictly proper and $W_1$ is proper (in the sensitivity
minimization problem) and, $W_1$ is proper and $W_2$ is improper (in
the mixed sensitivity minimization problem).
\end{lemma}
\begin{proof} Transfer function $F(s)$can be written as ratio of two
polynomials, $N_F$ and $D_F$, with degrees $m$ and $n$ respectively.
We can define relative degree function, $\phi$, as \bd
\phi(F(s))=\phi\left(\frac{N_F(s)}{D_F(s)}\right)=n-m.\ed Note that
$\phi(F_1(s)F_2(s))=\phi(F_1(s))+\phi(F_2(s))$ and
$\phi(F(s)F(-s))=2\phi(F(s))$.

The optimal controller has finitely many unstable poles if
$F_{\gamma_{opt}}$ is strictly proper, i.e.
$\phi(F_{\gamma_{opt}}(s))>0$. To show this, we can write by using
definition of $F_{\gamma_{opt}}$ and
(\ref{eq:spectralfactorization}), \bea
\nonumber \phi(F_{\gamma_{opt}}(s))&=&\phi(G_{\gamma_{opt}}(s)), \\
\nonumber &=&\frac{1}{2}\;\phi(\left(W_1(s)W_1(-s)+W_2(s)W_2(-s)-\gamma_{opt}^{-2}W_1(s)W_1(-s)W_2(s)W_2(-s)\right)^{-1}), \\
\nonumber &=&-\frac{1}{2}\;\phi(\left(W_1(s)W_1(-s)+W_2(s)W_2(-s)-\gamma_{opt}^{-2}W_1(s)W_1(-s)W_2(s)W_2(-s)\right)), \\
\nonumber &=&-\frac{1}{2}\;\min{\left\{\phi(W_1(s)W_1(-s)), \phi(W_2(s)W_2(-s)), \phi(W_1(s)W_1(-s)W_2(s)W_2(-s))\right\}}, \\
\nonumber
&=&-\min{\left\{\phi(W_1(s)),\phi(W_2(s)),\phi(W_1(s))+\phi(W_2(s))\right\}}.
\eea Strictly properness of $F_{\gamma_{opt}}$ implies, \be
\label{eq:mininequality}
\min{\left\{\phi(W_1(s)),\phi(W_2(s)),\phi(W_1(s))+\phi(W_2(s))\right\}}<0.
\ee We know that $\phi(W_1(s))\geq0$ and $\phi(W_2(s))\leq0$,
\cite{FOT}. Therefore, the inequality (\ref{eq:mininequality}) is
satisfied if and only if $\phi(W_1(s))\geq0$ and $\phi(W_2(s))<0$
are valid which means that $W_1(s)$ is proper and $W_2(s)$ is
improper. Since we have $(W_2 N_o)^{-1}\in\RHi$ \cite{FOT}, we can
conclude that the plant is strictly proper. Same proof is valid for
the suboptimal case.
\end{proof}

We know that the suboptimal $\Hi$ controllers are written as
(\ref{eq:Csubopt}). It is possible to rewrite the suboptimal
controllers as, \bd
C_{subopt}(s)=\frac{\left(\frac{N_o^{-1}(s)F_\rho(s)} {dE_\rho(s)~
dm_d(s)}\right)(L_2(s)+L_1(-s)m_n(s)F_\rho(s))} {P_1(s)+P_2(s) U(s)}
\ed where \bea \label{eq:P1_2} \nonumber
P_1(s)&=&\frac{L_1(s)+L_2(s)m_n(s)F_\rho(s)}
{nE_\rho(s)~ nm_d(s)}, \\
P_2(s)&=&\frac{L_2(-s)+L_1(-s)m_n(s)F_\rho(s)}{nE_\rho(s)~nm_d(s)},
\eea and $nE_\rho$, $dE_\rho$ and $nm_d$, $dm_d$ are minimal order
coprime numerator and denominator polynomials of
$E_\rho=\frac{nE_\rho}{dE_\rho}$ and $m_d=\frac{nm_d}{dm_d}$.

The unstable poles of $C_{subopt}$ are the $\rhp$ zeros of
$P_1+P_2U$. If there exists a $U\in\RHi$ with $\|U\|_\infty<1$, such
that $P_1+P_2U$ has no unstable zeros, then the corresponding
suboptimal controller is stable.

Assume that $F_\rho$ is strictly proper which implies $P_1$ and
$P_2$ has finitely many unstable zeros. The suboptimal controller is
stable if and only if $S_U:=(1+\tilde{P}U)^{-1}$ is stable where
$\tilde{P}=\frac{P_2}{P_1}$. Note that since $P_1$ and $P_2$ has
finitely many unstable zeros, we can write $\tilde{P}$ as, \bd
\tilde{P}=\frac{\tilde{M}}{\tilde{M_d}} \tilde{N_o} \ed where
$\tilde{M}$ and $\tilde{M_d}$ are inner, finite dimensional and
$\tilde{N_o}$ is outer and infinite dimensional. Finding stable
$S_U$ with $U\in\Hi$ is considered as sensitivity minimization
problem with stable controller, \cite{Ganesh}. However, $U$ has a
norm restriction as $\|U\|_\infty\leq 1$ in our problem. Note that
$U$ can be written as, \bd U(s)=\left(\frac{1-S_U(s)}{S_U(s)}\right)
\left(\frac{P_1(s)}{P_2(s)}\right). \ed

Define $\mu_{opt}$ as, \bd \mu_{opt}=\inf_{U\in\Hi}\|S_U\|_\infty=
\inf_{U\in\Hi}\|(1+\tilde{P}U)^{-1}\|_\infty. \ed If we fix $\mu$ as
$\mu>\mu_{opt}$, then there exists a free parameter $Q$ with
$\|Q\|_\infty\leq1$ which parameterizes all functions stabilizing
$S_U$ and achieving performance level $\mu$. The notation for the
sensitivity function achieving performance level $\mu$ is
$S_{U,\mu}(Q)$.

\begin{lemma} \label{lemma:w1w2finmany}
Assume that the weights in mixed sensitivity minimization problem
(\ref{eq:wsm}), $W_1$ and $W_2$, are proper and improper
respectively and $\mu_o>\mu_{opt}$. If there exists  $Q_o$ with
$\|Q_o\|_\infty\leq1$ satisfying \be
\left|\left(\frac{1-S_{U,\mu_o}(Q_o(j\omega))}
{S_{U,\mu_o}(Q_o(j\omega))}\right)
\left(\frac{P_1(j\omega)}{P_2(j\omega)}\right)\right| \leq 1, \ee
then the suboptimal $\Hi$ controller, $C_{subopt}$, is stable and
achieves the performance level $\rho$ by selecting the parameter $U$
as, \be U(s)=\left(\frac{1-S_{U,\mu_o}(Q_o(s))((s)}
{S_{U,\mu_o}(Q_o(s))}\right)\left(\frac{P_1(s)}{P_2(s)}\right) \ee
\end{lemma}
\begin{proof}
The result of Lemma is immediate. Since $Q_o$ satisfies the norm
condition of $U$ and makes $S_{U,\mu}(Q_o)$ stable, the suboptimal
controller has no right half plane poles by selection of $U$ as
shown in theorem.
\end{proof}

There is no need to search for $\mu_{opt}$, since $U$ has always an
essential singularity at infinity for the optimal case, see
\cite{Ganesh}. By a numerical search, we can find $Q_o$ satisfying
the norm condition for $U$. Instead of finding $U$ resulting in a
suboptimal stable controller, the problem is transformed into
finding $Q_o$ satisfying the norm condition. First problem needs to
check whether a quasi-polynomial has unstable zeros. By Lemma
\ref{lemma:w1w2finmany}, this problem is reduced into stable
function search with infinity norm less than $1$ and a norm
condition for $U$. Conservatively, the search algorithm for $Q_o$
can be done for first order bi-proper functions such that
$Q_o(s)=u_\infty\left(\frac{s+z_u}{s+p_u}\right)$ where $p_u>0$,
$z_u\in\R$, and $|u_\infty|\leq\max{\{1,\frac{p_u}{|z_u|}\}}$. The
algorithm for this approach is explained below.

\noindent \textbf{Algorithm}\\
Assume that the optimal and central suboptimal controllers have
finitely many unstable poles. We can design a stable suboptimal
$\Hi$ controller by the following algorithm:

\begin{enumerate}
  \item[1)] Fix $\rho>\gamma_{opt}$,
  \item[2)] Obtain $P_1$ and $P_2$. If $P_1$ has no unstable
  zero, then suboptimal controller is stable for $U=0$. If
  not, go to step 3.
  \item[3)] Define the right half plane zeros of $P_1$ and $P_2$
  as $\{p_i\}_{i=1}^{n_p}$ and $\{s_i\}_{i=1}^{n_s}$ respectively.
  Define $\tilde{M}_d(s)$ and $\tilde{M}(s)$ as
  \be \label{eq:MtdMs}
  \tilde{M}_d(s)=\prod_{i=1}^{n_p}\frac{s-p_i}{s+p_i}, ~~~~~~
  \tilde{M}(s)=\prod_{i=1}^{n_s}\frac{s-s_i}{s+s_i}
  \ee
  and calculate
  \be \label{eq:wz}
   w_i=\left(\tilde{M_d}(s_i)\right)^{-1}, ~~~~~~
  z_i=\frac{s_i-a}{s_i+a}, \quad i=1,\ldots,n_s ~~~~\mbox{where
  $a>0$.}
  \ee
  \item[4)] Search for minimum $\mu$ which makes the Pick matrix
  positive semi-definite,\be \label{eq:pickmatrix}
  Q_{{P}_{(i,k)}}^{\mu}=\frac{\ln(\frac{\mu^2}
  {w_i\bar{w}_k})+j2\pi(n_k-n_i)}{1-z_i\bar{z}_k}
  \ee
  where $Q\in\mathbb{C}^{n_s\times n_s}$ and $n_{[.]}$ is integer.
  Note that most of the integers will
  not result in positive semi-definite Pick matrix. Therefore, for
  each integer set, we can find the smallest $\mu$ and $\mu_{opt}$
  will be the minimum of these values. For details, see
  \cite{Ganesh}.
  \item[5)] Fix $\mu$ such that $\mu>\mu_{opt}$. For all
  possible integer set, obtain $g(z)\in\Hi$ with interpolation
  conditions, \be \label{eq:gzinterpcond}
  g(z_i)=-\ln{\frac{w_i}{\mu}}-j2\pi n_i. \ee Note that since
  $g(z)$ has a free parameter $q(z)$ with $\|q\|_\infty\leq1$,
  we can write the function as
  $g_q(z)$. Then, search for parameters ($u_\infty$, $z_u$, $p_u$)
  satisfying
  \be \label{eq:step5magcond}
  \max_{{
  \omega\in[0,\infty)}}{\left|\frac{(1-S_{U,\mu}(j\omega))}
  {\frac{S_{U,\mu}(j\omega)P_2(j\omega)}{P_1(j\omega)}}\right|}\leq1,
  \ee
  where
  \bea \label{eq:SU}
  S_{U,\mu}(s)&=&\mu\tilde{M}_d(s)e^{-G_Q(s)}, \\
\nonumber  G_Q(s)&=&g_q\left(\frac{s-a}{s+a}\right)
  \eea and $Q(s)=u_\infty\left(\frac{s+z_u}{s+p_u}\right)$ as defined
  before. If one of the parameter set satisfies the inequality,
  then $Q_o=u_{\infty,o}\left(\frac{s+z_{u,o}}{s+p_{u,o}}\right)$
  and corresponding $U$ results in a stable suboptimal $\Hi$ controller, stop.
  If no parameter set satisfies the inequality,
  repeat the procedure for sufficiently high $\mu$,
  until a pre-specified maximum is reached, go next
  step.
  \item[6)] Increase $\rho$, go to step 2,
  if a maximum pre-specified $\rho$ is reached, stop.
  This method fails to provide a stable $\Hi$ controller.
\end{enumerate}

An illustrative example is presented in Section
\ref{subsec:exampfin}.

\section{Examples} \label{sec:examp}

Two examples will be given in this section. In the first example,
the optimal and central suboptimal controllers have infinitely many
unstable poles. By using the design method in Section
\ref{sec:infinitelymany}, we show that there exists a stable
suboptimal controller even the magnitude condition in
(\ref{eq:magcond}) is violated for low frequencies. In other words,
the example illustrates that the conditions in Theorem
\ref{eq:thmss} are only sufficient.

The second example explains the design method for suboptimal stable
$\Hi$ controller when central controller has finitely many unstable
poles. The algorithm is applied step by step as given in Section
\ref{sec:finitelymany}.

\subsection{Example with Infinitely Many Unstable Poles}
\label{subsec:exampinf} Let the weight functions in mixed
sensitivity problem (\ref{eq:wsm}) be $W_1(s)=\frac{1+0.1s}{0.4+s}$
and $W_2=0.5$, and consider the plant \be \label{example:IFtf}
P(s)=\frac{r_{p}(s)}{t_{p}(s)}=\frac{\sum_{i=1}^2
r_{p,i}(s)e^{-h_is}} {\sum_{i=1}^3 t_{p,i}(s)e^{-\tau_is}}
=\frac{(s+3)+2(s-1)e^{-0.4s}}{s^2+se^{-0.2s}+5e^{-0.5s}}. \ee The
denominator of the plant, $t_p(s)$ has finitely many $\rhp$ zeros at
$0.4672\pm1.8890j$, whereas $r_p(s)$ has infinitely many $\rhp$
zeros converging to $1.7329\pm j(5k+2.5)\pi$ as
$k\rightarrow\infty$, $k\in\mathbb{Z}_+$. The plant satisfies
assumptions $A.1$-$A.2$. We can rewrite the plant $P$ in the form
(\ref{eq:seqplant}) where $n=2$, $m=3$, \bea
 \nonumber R_i(s)&=&\frac{r_{p,i}(s)}{(s+1)^2},
 \quad\textrm{and}\quad T_j(s)=\frac{t_{p,j}(s)}{(s+1)^2}.
\eea One can see that $R$ is an $I$-system whose conjugate
$\bar{R}=-\frac{2(s+1)+(s-3)e^{-0.4s}}{(s+1)^2}$ has only one $\rhp$
zero, $0.247$ and $T$ is an $F$-system with two $\rhp$ zeros,
$0.465\pm 1.890j$. Therefore, assumptions $A.3$-$A.4$ are satisfied
by Corollary~\ref{cor:A3A4} and the plant $P$ can be factorized as
(\ref{eq:mnfac}) using (\ref{eq:IFfactorization}) \bea
\label{eq:IFexamplefactorization} \nonumber
m_n&=&M_{\bar{R}}\frac{R}{\bar{R}}
=\left(\frac{s-0.247}{s+0.247}\right)
\frac{\left(\frac{(s+3)+2(s-1)e^{-0.4s}}{(s+1)^2}\right)}
{\left(\frac{2(s+1)+(s-3)e^{-0.4s}}{(s+1)^2}\right)}, \\
\nonumber m_d&=&M_T=\left(\frac{s^2-0.93s+3.79}{s^2+0.93s+3.79}\right), \\
N_o&=&\frac{\bar{R}}{M_{\bar{R}}} \frac{M_T}{T} \eea where
$T=\left(\frac{s^2+se^{-0.2s}+5e^{-0.5s}}{(s+1)^2}\right)$, $N_o$ is
outer, $m_n$, $m_d$ are inner functions, infinite and finite
dimensional respectively. For details, see \cite{GO-IFAC-06}.

From \cite{FOT}, the optimal performance level is
$\gamma_{opt}=0.57$. The optimal controller has infinitely many
$\rhp$ poles converging to $s=0.99\pm j(5k+2.5)\pi$ as
$k\rightarrow\infty$, $k\in\mathbb{Z}_+$. If central suboptimal
controller (i.e., $U=0$) is calculated for $\rho=0.67$, it has
infinitely many $\rhp$ poles converging to $s=0.37\pm j(5k+2.5)\pi$
as $k\rightarrow\infty$, $k\in\mathbb{Z}_+$. The suboptimal
controllers can be written as (\ref{eq:Csubopt}) where
\bea
\nonumber E_\rho&=&\frac{0.93+0.44s^2}{0.45(0.16-s^2)},\\
\nonumber F_\rho&=&0.67\left(\frac{0.4-s}{0.70+0.50s}\right), \\
\nonumber L_2&=&0.79s^3+2.51s^2+2.84s+3.43,\\
\nonumber L_1&=&s^3+1.49s^2+1.86s+0.65. \eea

We will use the design method of Section \ref{sec:infinitelymany} to
find a stable suboptimal controller by search for $U$ such that
$\|U\|_\infty\leq1$.  For simplicity, the algorithm is tried for the
case, $U(s)=u_\infty$.
\begin{enumerate}
  \item[1)] Fix $\rho=0.67 > \gamma_{opt}=0.57$,
  \item[2)] $k=0.79$ and $f_\infty=1.33$ are calculated.
  \item[3)] $n_1=1$, $\ell=2$, $n_1+\ell$ is odd and $|k|<1$. By using Lemma
  \ref{eq:inftyvalues}, the admissible interval for $u_\infty$ is $(0.095,0.96)$.
  \item[4)] $L_{1U}(s)$ is stable for $u_\infty\in(-0.19,0.46)$.
  \item[5)] Overall admissible values for $U$ are $u_\infty\in(0.095,0.46)$.
  The values of $\omega_{max}$ and $\eta_{max}$ for all admissible
  $u_\infty$ range can be  seen  in Figure \ref{fig:wetamax}.
  One can minimize both $\omega_{max}$ and $\eta_{max}$
  by finding the intersection of two curves,
  i.e.
  \bd
  u^o_{\infty}={\rm arg}~\min_{u_\infty} \max\{
  \omega_{max},\eta_{max}\}=0.35.
  \ed
  \item[6)] One can see that Nyquist plot in clockwise direction of $m_n F_\rho L_U$ encircles $-1$ twice in clockwise
  direction.
  Note that the unstable zeros of $E_\rho(s)$ and
  $m_d$ are $\pm1.45j$, $0.47\pm1.89j$, respectively. Since the zeros
  on the imaginary axis are excluded from Nyquist plot, there are no
  unstable zeros of $1+m_n F_\rho L_U$.
\end{enumerate}
Therefore, we can conclude that suboptimal controller is stable for
$U(s)=0.35$ and achieves the $\Hi$ norm $\rho=0.67$. For practical
implementation, the suboptimal controller found can be represented
as cascade and feedback connections containing finite impulse
response filter that does not have unstable pole-zero cancellations
in the controller, as explained in \cite{GO-IFAC-06}.
\begin{figure}[h]
\begin{center}
\includegraphics[width=8.5cm]{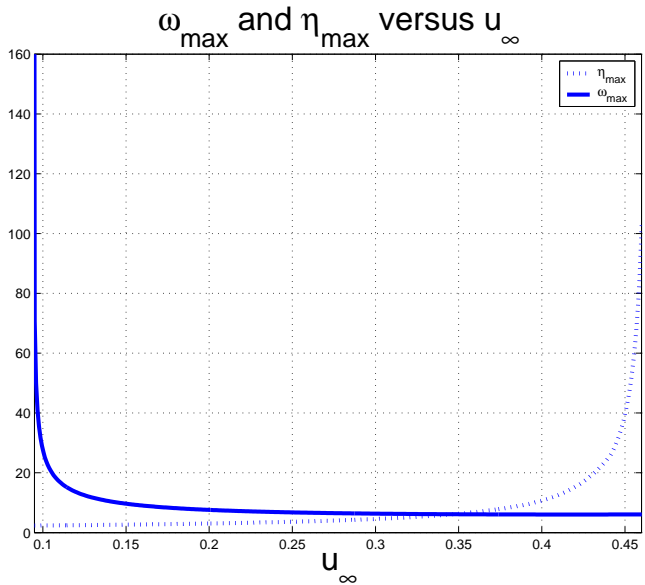}
\caption{$w_{max}$ and $\eta_{max}$ versus $u_\infty$}
\label{fig:wetamax}
\end{center}
\end{figure}
\subsection{Example with Finitely Many Unstable Poles}\label{subsec:exampfin}
For the  plant (\ref{example:IFtf}) and weights
$W_1(s)=\left(\frac{1+0.1s}{0.4+s}\right)$ and $W_2(s)=(0.01s+0.5)$,
we find the optimal performance level as $\gamma_{opt}=0.59$. The
corresponding optimal $\Hi$ controller can be written as
(\ref{eq:Copt}) which has unstable poles at $0.67\pm14.09j$,
$0.11\pm28.33j$. Note that all suboptimal $\Hi$ controllers for
finite dimensional $U$ will have finitely many unstable poles by
Corollary~\ref{cor:infitycond}. Therefore we can apply the algorithm
in Section~\ref{sec:finitelymany}.

\begin{enumerate}
  \item[1)] Fix $\rho=0.60>\gamma_{opt}=0.59$,
  \item[2)] The suboptimal controllers can be written as in (\ref{eq:Csubopt})
where $m_n$ is given in (\ref{eq:IFexamplefactorization}) and \bea
\nonumber E_\rho&=&\frac{0.94+0.35s^2}{0.36(0.16-s^2)},\\
\nonumber F_\rho&=&\frac{0.36(0.4-s)}{0.0059s^2+0.31s+0.35}, \\
\nonumber L_2&=&0.98s^3+2.45s^2+1.91s+2.10,\\
\nonumber L_1&=&s^3+1.64s^2+0.45s+1.61, \eea and $U$ is a free
parameter such that $U\in\Hi$, $\|U\|_\infty\leq1$. We can obtain
$P_1$ and $P_2$ from (\ref{eq:P1_2}). Note that $P_1$ has $\rhp$
zeros at $p_{1,2}=0.64\pm14.064j$, $p_{3,4}=0.081\pm28.314j$ and
$P_2$ has $\rhp$ zeros at $s_{1,2}=0.29\pm28.31j$,
$s_{3,4}=0.90\pm14.035j$ and $s_5=2.43$. Therefore, the central
controller (when $U=0$) for the chosen performance level,
$\rho=0.6$, is unstable.
  \item[3)] Note that $\rhp$ zeros of $P_1$ and $P_2$ are defined in
  previous step. Then, $\tilde{M}_d$ and $\tilde{M}$ can be defined
  as (\ref{eq:MtdMs}) where $n_s=4$ and $n_p=5$.
  By (\ref{eq:wz}), $w_i$ and $z_i$ can be calculated where
  conformal mapping parameter, $a$, is chosen as $1$.
  \item[4)] For all possible integers sets,
  the minimum $\mu$ resulting in positive semi-definite
  Pick matrix (\ref{eq:pickmatrix}), is $\mu_{opt}=6.15$ in which all
integers are equal to $0$.
  \item[5)] Fix $\mu=100$. The interpolation
  conditions for $g(z)$ can be written as in
  (\ref{eq:gzinterpcond}) where all integers, $n_i$, are zero.
  By the Nevanlinna-Pick interpolation, (see e.g.\cite{FOT,ZO98}),
  $g_q(z)$ is obtained.
  By transformation, $G_Q(s)$ can be calculated where $Q(s)$
  is a parameterization term such that $Q\in\Hi$ and
  $\|Q\|_\infty\leq1$. We will search for $Q$ satisfying the
  inequality (\ref{eq:step5magcond}) in the form of
  $Q(s)=u_\infty$ with $|u_\infty|\leq1$.
  Note that we choose $z_u=p_u=0$ and all
  functions in (\ref{eq:SU}) and $P_1$, $P_2$  are
  defined before. The search shows that (\ref{eq:step5magcond})
  is satisfied for $u_\infty\in[0.23,0.33]$.
  The  magnitude of $U(j\omega)$ is shown for $u_\infty=0.3$
  in Figure \ref{fig:Uversusw}.  Note that  $\|U\|_\infty\leq 1$.
  As a result, stable $\Hi$ controller achieves the performance
  level, $\rho=0.6$.
  By a numerical search, we can find many $u_\infty$ values for
  different $\mu$ resulting in stable $\Hi$ controller at $\rho=0.6$
  provided  that $U$ satisfies the norm condition for chosen $Q=u_\infty$.
  The various $u_\infty$ values resulting stable $\Hi$ controller can be seen
  in Figure \ref{fig:Uinfversusuinf}. We observe that as $\mu$
  is increased, the range of $u_\infty$ stabilizing the controller
  decreases.
\end{enumerate}

\begin{figure}[h]
\begin{minipage}[t]{0.5\linewidth}
\begin{flushright}
\includegraphics[width=8.5cm]{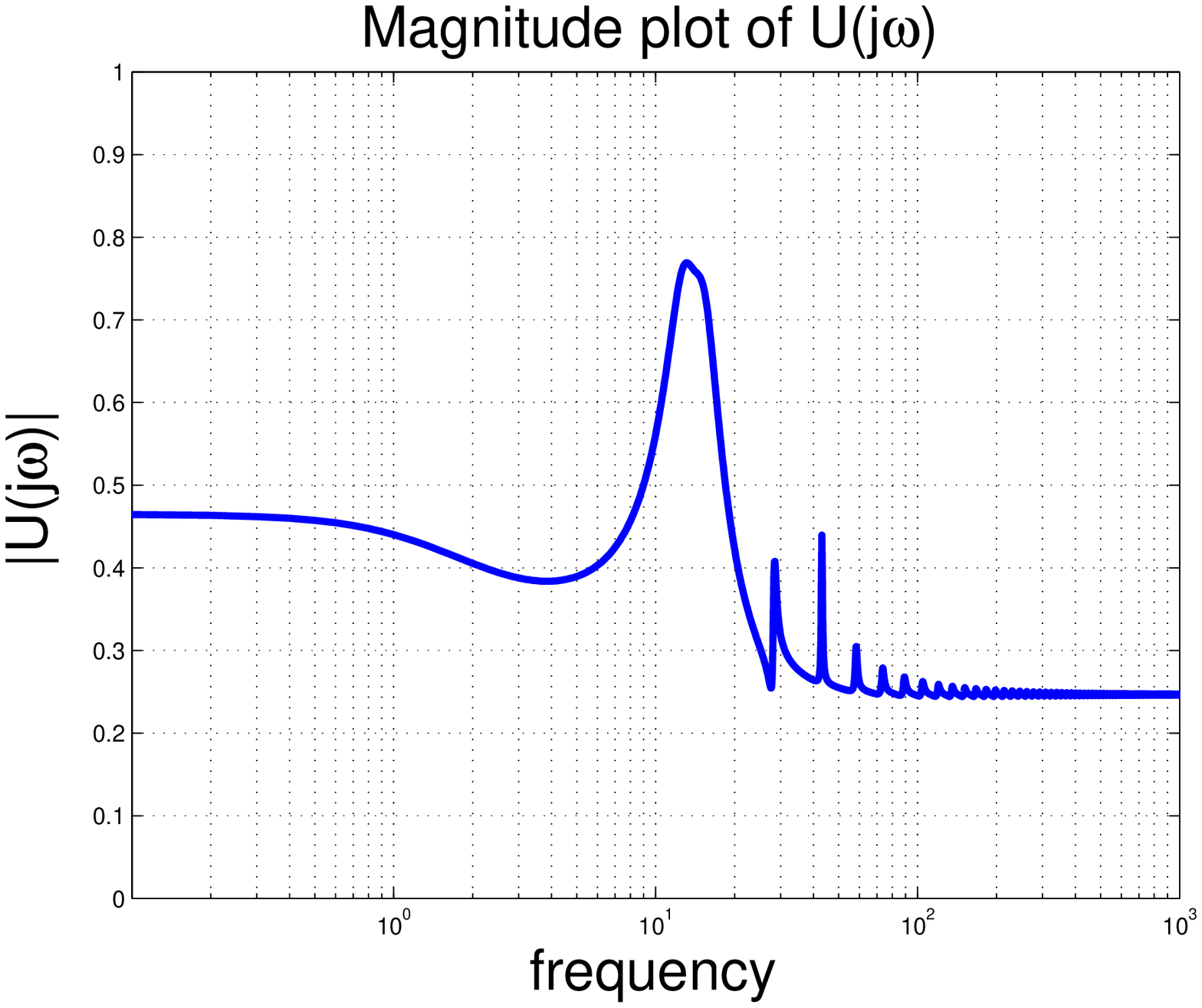}
\caption{$|U(j\omega)|$ for $\mu=100$ and $u_\infty=0.3$}
\label{fig:Uversusw}
\end{flushright}
\end{minipage}
\hfill
\begin{minipage}[t]{0.5\linewidth}
\begin{flushleft}
\includegraphics[width=8.5cm]{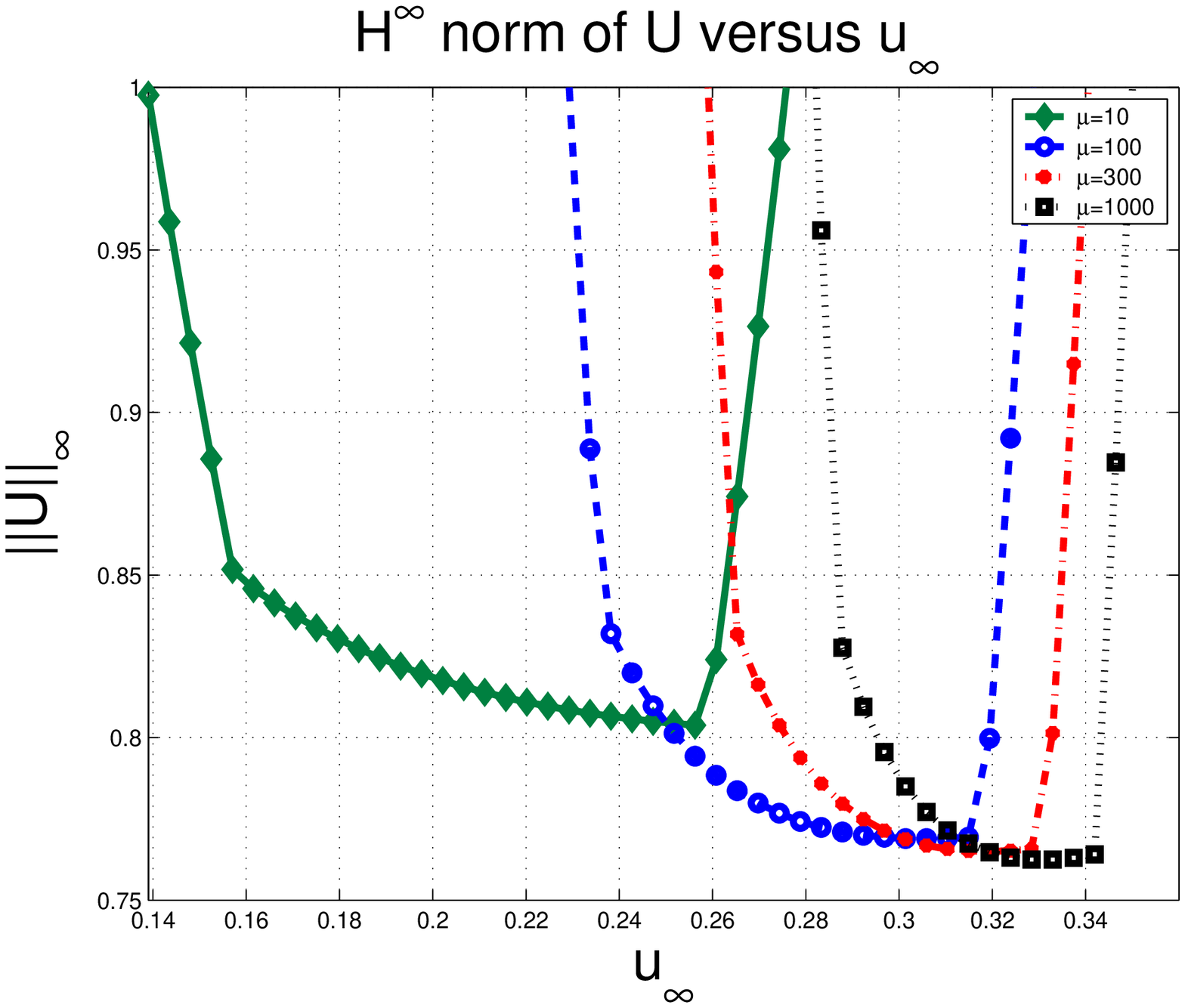}
\caption{Feasible values of $u_\infty$} \label{fig:Uinfversusuinf}
\end{flushleft}
\end{minipage}
\end{figure}

\section{Conclusions} \label{sec:concl}
In this paper, stability of $\Hi$ controllers are investigated for
general time-delay systems. Conditions on the problem data (plant
and the weights) are derived that make the optimal and central
suboptimal controllers unstable, with finitely or infinitely many
$\rhp$ poles.
A search method is proposed for finding stable suboptimal
controllers by properly selecting the free design parameter $U$
appearing in the parameterization of all suboptimal $\Hi$
controllers for the class of time delay systems considered. When the
optimal and central suboptimal controllers have finitely many $\rhp$
poles the search algorithm uses the Nevanlinna-Pick interpolation to
derive feasible parameters of the first order $U$. When the optimal
and central suboptimal controllers have infinitely many poles in
$\rhp$, the search algorithm uses a Nyquist argument at each step.

\small{
}
\end{document}